\documentclass{llncs} 
\usepackage{comment}
\usepackage{amsmath,amssymb}
\usepackage{color}
\usepackage{graphicx}
\usepackage{epstopdf}
\usepackage{epsfig}
\usepackage{float}

\newcommand{\NP}{\ensuremath{\mathsf{NP}}}
\newcommand{\PPAD}{\ensuremath{\mathsf{PPAD}}}
\newcommand{\PLS}{\ensuremath{\mathsf{PLS}}}
\newcommand{\TFNP}{\ensuremath{\mathsf{TFNP}}}
\newcommand{\FIXP}{\ensuremath{\mathsf{FIXP}}}
\newcommand{\FIXPA}{\ensuremath{\mathsf{FIXP}_a}} 
\def\nat{{\mathbb N}}
\def\int{{\mathbb Z}}
\def\real{{\mathbb R}}
\def\rat{{\mathbb Q}}

\begin{document}

\title{The complexity of approximating a trembling hand perfect equilibrium 
of a multi-player game in strategic form}

\author{Kousha Etessami \inst{1} \and Kristoffer Arnsfelt Hansen \inst{2} \and Peter Bro Miltersen\inst{2} \and Troels Bjerre S{\o}rensen \inst{3}}

\institute{
University of Edinburgh. \email{kousha@inf.ed.ac.uk} \and
Aarhus University. \email{\{arnsfelt,bromille\}@cs.au.dk} \and
IT-University of Copenhagen. \email{trbj@itu.dk}
}

\maketitle

\begin{abstract}
 We consider the task of computing an approximation of a trembling
  hand perfect equilibrium for an $n$-player game in strategic
  form, $n \geq 3$. We show that this task is complete for the complexity class
  \FIXPA. In particular, the task is polynomial time equivalent to the
  task of computing an approximation of a Nash equilibrium 
in strategic form games with three (or more) players.
\end{abstract}

\section{Introduction}
Arguably \cite{vanDamme:1991:SPNE}, the most important {\em refinement} 
of Nash equilibrium for finite games in strategic
form (a.k.a. games in normal form, i.e., games given by their tables of payoffs)  is Reinhard Selten's \cite{Selten:1975:PCEE} notion of 
{\em trembling hand perfection}. The set of trembling hand perfect
equilibria of a game is a non-empty subset of the Nash equilibria of that game.
Also, many ``unreasonable'' Nash equilibria of many games, e.g., those relying on ``empty threats" in equivalent extensive forms of those games, are not trembling hand
perfect, thus motivating and justifying the notion. The importance of the notion is illustrated by the fact that Selten received the Nobel prize in economics together with Nash (and Harsanyi), ``for their pioneering analysis of equilibria in the theory of non-cooperative games". In this paper, we study the computational complexity of finding trembling hand perfect equilibria of games given in strategic form.

The computational complexity of finding a {\em Nash} equilibrium of a
game in strategic form is well-studied. When studying this
computational task, we assume that the game given as input is
represented as a table of integer (or rational) payoffs, with each
payoff given in binary notation. The output is a strategy profile,
i.e., a family of probability distributions over the strategies of
each player, with each probability being a rational number with
numerator and denominator given in binary notation. The computational
task is therefore discrete and we are interested in the Turing machine
complexity of solving it. Papadimitriou \cite{Papadimitriou:1990:GTL}
showed that for the case of two players, the problem of computing an
{\em exact} Nash equilibrium is in \PPAD, a natural complexity class
introduced in that paper, as a consequence of the Lemke-Howson
algorithm \cite{Lemke:1965:BEP} for solving this task.  For the case
of three or more players, there are games where no Nash equilibrium which uses
only rational probabilities exists \cite{Nash:1951:NCG}, and hence
considering some relaxation of the notion of ``computing'' a Nash equilibrium is
necessary to stay within the discrete input/output framework outlined above. In
particular, Papadimitriou showed that the problem of computing an
$\epsilon$-Nash equilibrium, with $\epsilon > 0$ given as part of the
input in binary notation, is also in \PPAD, as a consequence of 
Scarf's algorithm \cite{Scarf67} for solving this task. Here, an
$\epsilon$-Nash equilibrium is a strategy profile where no player can
increase its utility by more than $\epsilon$ by deviating. In breakthrough
papers, Daskalakis {\em et al.} \cite{Daskalakis:2006:CCNE} and Chen
and Deng \cite{Chen:2006:SCTN} showed that both tasks are also hard
for \PPAD, hence settling their complexity: Both are
\PPAD-complete. Subsequently, Etessami and Yannakakis \cite{EY07}
pointed out that for some games, $\epsilon$-Nash equilibria can be so
remote from any exact Nash equilibrium (unless $\epsilon$ is so small
that its binary notation has encoding size exponential in the size of the
game), that the former tells us little or nothing about the latter. For
such games, the $\epsilon$-Nash relaxation is a bad proxy for Nash
equilibrium, assuming the latter is what we are actually interested
in computing. Motivated by this, they suggested a different
relaxation: Compute a strategy profile with $\ell_\infty$-distance at
most $\delta$ from an exact Nash equilibrium, with $\delta > 0$ again
given as part of the input in binary notation. In other words, {\em
  compute an actual Nash equilibrium to a desired number of
  bits of accuracy}. They showed that this problem is complete for a
natural complexity class \FIXPA\ that they introduced in the same
paper. 
Informally, \FIXPA\  
is the class of discrete search problems that can be
reduced to approximating (within desired $\ell_\infty$-distance) 
any one of the Brouwer fixed points of a function given
by an algebraic circuit using gates: $+,-,*,/,\max,\min$.
(We will formally define \FIXPA\ later.)

In this paper, we want to similarly understand the case of trembling hand perfect equilibrium. For the case of two players, the problem of computing an exact trembling hand perfect equilibrium is \PPAD-complete. 
This follows from a number of known exact pivoting algorithms for computing refinements of this notion \cite{vonStengel:2002:CNFP,Miltersen:2008:CQE,TroelsEC}.  For the case of three or more players, we are not aware of any natural analogue of the notion of $\epsilon$-Nash equilibrium as an approximate proxy for a trembling hand perfect equilibrium.\footnote{The 
already studied notion of an $\epsilon$-perfect equilibrium ($\epsilon$-PE), which we 
discuss later, 
does {\em not} qualify as such an analogue: For some three-player games, every $\epsilon$-PE uses 
irrational probabilities, and thus  ``computing'' an (exact) $\epsilon$-PE is just as problematic as computing an exact NE.
Indeed, the notion of a $\epsilon$-PE is used as a technical step towards the definition of 
trembling hand perfect equilibrium, 
rather than as a natural ``numerical relaxation" of this notion.}
Thus, we only discuss in this paper the approximation notion of Etessami and Yannakakis. The main result of the present paper is the following:

\begin{theorem}\label{thm-main}The following computational task is $\FIXPA$-complete for any $n \geq 3$: Given an integer payoff table 
for an $n$-player game $\Gamma$, and a rational $\delta > 0$, with all numbers given in standard binary notation, compute (the binary representation of) a strategy profile $x'$ with rational probabilities having 
$\ell_\infty$ distance at most $\delta$ to a
trembling hand perfect equilibrium 
of $\Gamma$.
\end{theorem}

As an immediate corollary of our main theorem, and the results of
Etessami and Yannakakis, we have that {\em approximating a Nash
  equilibrium and approximating a trembling hand perfect equilibrium
  are polynomial time equivalent tasks}. In particular, there is a
polynomial time algorithm that finds an approximation to a trembling
hand perfect equilibrium of a given game, using access to any oracle
solving the corresponding approximation problem for the case of Nash
equilibrium. To put this result in perspective, we note that Nash
equilibrium and trembling hand perfect equilibrium are computationally
quite different in other respects: if instead of {\em finding} an
equilibrium, we want to {\em verify} that a given strategy profile is
such an equilibrium, the case of Nash equilibrium is trivial, while
the case of trembling hand perfect equilibrium for games
with 3 (or more) players is \NP-hard
\cite{HMS10}.  This might lead one to believe that approximating a
trembling hand perfect equilibrium for games
with 3 or more players is likely to be harder than
approximating a Nash equilibrium, but we show that this is not the
case.
\subsection{About the proof}

Informally (for formal definitions, see below), \FIXP\ (resp., \FIXPA)
is defined as the complexity class of search problems that can be cast
as exactly computing (resp., approximating) a {\em Brouwer fixed
  point} of functions represented by circuits over basis
$\{+,*,-,/,\max,\min \}$ with rational constants. It was established
in \cite{EY07} that computing (resp., approximating) an actual Nash
Equilibrium (NE) for a finite $n$-player game is \FIXP-complete
(resp., \FIXPA-complete), already for $n = 3$. Since trembling hand
perfect equilibria constitute a refinement of Nash Equilibria, to
show that approximating a trembling hand perfect equilibrium is
$\FIXPA$-complete, we merely have to show that this task is in
$\FIXPA$.

An $\epsilon$-perfect equilibrium ($\epsilon$-PE for short) is defined to be a fully mixed strategy profile, $x$, where every strategy $j$ of every player $i$ that is played with  probability $x_{i,j} > \epsilon$ must be a best response to the other player's strategies
$x_{-i}$. Then, a trembling hand perfect equilibrium (PE for short) is defined to be a limit point of a sequence of
$\epsilon$-PEs, for $\epsilon > 0$, $\epsilon \rightarrow 0$.
Here, by limit point we mean, as usual, any point to which a subsequence
 of the sequence converges. Such a point must exist, by the Bolzano-Weierstrass theorem.

In rough outline, 
our proof that approximating a PE is in $\FIXPA$ 
has the following structure:

\begin{enumerate}

\item  We first define (in section \ref{Kousha}) for any $n$-player game $\Gamma$,
a map, 
$F^{\epsilon}_{\Gamma}$,
parameterized by a 
parameter $\epsilon > 0$, so that  $F^{\epsilon}_{\Gamma}$ defines
a map from $D^{\epsilon}_{\Gamma}$ to itself, where $D^{\epsilon}_{\Gamma}$ denotes the space 
of fully mixed strategy strategy profiles $x$ 
such that every player 
plays each strategy with probability at least $\epsilon$. Also,
$F^{\epsilon}_{\Gamma}(x)$ is described by a $\{+,-,*,\min,\max\}$-circuit
with $\epsilon$ as one of its inputs.
 In particular, the Brouwer fixed point theorem applies to this map.
We show that the circuit defining $F^{\epsilon}_{\Gamma}$ can
be computed in polynomial time from the input game instance $\Gamma$, and
that  every Brouwer fixed point of $F^{\epsilon}_{\Gamma}$
is an $\epsilon$-PE of the original game $\Gamma$, making crucial use of, and
modifying, 
a new fixed point
characterization of NEs that was defined and used in \cite{EY07}.
\item 
We then show (in section \ref{Kristoffer}) that if $\epsilon^* > 0$ is made sufficiently small as a function
of the encoding size $| \Gamma |$ of the game $\Gamma$,
and of a parameter $\delta > 0$,  
specifically if $\epsilon^* \leq \delta^{2^{g(|\Gamma|)}}$, 
where $g$ is some polynomial, then any $\epsilon^*$-PE
must be $\delta$-close (in the $l_\infty$-norm) to an actual PE. This part of the proof relies on real algebraic geometry.
\item
We then observe (in section \ref{wrapup}) that for any desired
$\delta$,
we can encode such a sufficiently small $\epsilon^* > 0$
as a circuit that is polynomially large in the encoding
size of $\Gamma$ and $\delta$, simply by {\em repeated squaring}. We think of this as constructing a {\em virtual infinitesimal} and believe that this technique will have many other applications in the context of proving $\FIXPA$ membership using real algebraic geometry.
Finally, plugging in the circuit for $\epsilon^*$ for the input $\epsilon$ in the circuit for $F^{\epsilon}_{\Gamma}$, we obtain a Brouwer function  
$F^{\epsilon^*}_{\Gamma}(x)$, defined by a $\{+,-,*,\max,\min\}$-circuit,
such that any fixed point of $F^{\epsilon^*}_{\Gamma}(x)$ is guaranteed
to be a fully mixed strategy profile, $x^*_{\epsilon^*}$, that is 
also within $l_\infty$ distance $\delta$ of a PE, $x^*$, of $\Gamma$. The triangle inequality completes the proof.

\end{enumerate}

\section{Definitions and Preliminaries}

\label{sec:definitions}

\subsection{Game-theoretic notions}
We use $\rat_+$ to denote the set of positive rational numbers. 
A finite $n$-player normal form game, 
$\Gamma = (N, \langle S_i \rangle_{i \in N}, \langle  u_i \rangle_{i \in N} )$,
consists of a set $N = \{1,\ldots,n\}$ of $n$ players indexed by their
number, a set of $n$ (disjoint) finite sets of {\em pure strategies},
$S_i$, one for each player $i \in N$,  
and $n$ rational-valued {\em payoff functions} $u_i : S \rightarrow \rat$,
from the product strategy space $S = \Pi^n_{i=1} S_i$ to $\rat$.

The elements of $S$, i.e., combinations of pure
strategies, one for each player, are called {\em pure strategy profiles}.
The assumption of rational values is for computational purposes.
Each rational number $r$ is represented as usual by its 
numerator and denominator in binary,
and we use $size(r)$ to denote the number of bits in the representation.
The size $|\Gamma |$ of the instance (game) $\Gamma$ is the total number of bits
needed to represent all the information in the game: the strategies of
all the players and their payoffs for all $s \in S$.

A {\em mixed strategy}, $x_i$, for a player $i$ is a probability distribution
on its set $S_i$ of pure strategies. Letting $m_i=|S_i|$, 
we view $x_i$ as a real-valued vector $x_i = (x_{i,1}, \ldots, x_{i,m_i})  
\in [0,1]^{m_i} $,  where $x_{i,j}$ denotes the probability with 
which player $i$ plays pure strategy $j$ in the mixed strategy $x_i$.
Note that we must have  $x_i \geq 0$ and $\sum^{m_i}_{i=1} x_{i,m_i} = 1$. That is,
a vector $x_i$ is  a mixed strategy of player $i$
iff it belongs to the {\em unit simplex} 
$\Delta_{m_i} = \{ y \in R^{m_i} | y \geq 0; \sum_{j=1}^{m_i} y_j =1 \}$.
We use the notation $\pi_{i,j}$ to identify
the pure strategy $j$ of player $i$, as well as its representation 
as a mixed strategy that assigns probability 1 to strategy $j$ and probability 
0 to the other strategies of player $i$. 

A {\em mixed strategy profile} $x = (x_1,\ldots,x_n)$ 
is a combination of mixed strategies
for all the players. That is, vector $x$ is a mixed strategy profile
iff it belongs to the product of the $n$ unit simplexes
for the $n$ players, 
$\{ x \in R^{m} \mid x \geq 0;  \sum_{j=1}^{m_i} x_{i,j} =1$ for 
$i=1,\ldots,k \}$.
We let $D_{\Gamma }$ denote the set of all mixed profiles
for game $\Gamma$.
The profile is {\em fully mixed}
if all the pure strategies of all players have nonzero probability.
We use the notation $x_{-i}$  to denote the subvector of $x$ induced by
the pure strategies of all players except for player $i$.
If $y_i$ is a mixed strategy of player $i$,
we  use $(y_i;x_{-i})$ to denote the mixed profile where 
everyone plays the same strategy as $x$ except for player $i$, who plays
mixed strategy $y_i$. 

The payoff function of each player can be extended from pure strategy profiles
to mixed profiles, and we will use $U_i$ to
denote the expected payoff function for player $i$.
Thus the (expected) payoff $U_i(x)$ of mixed profile $x$ for player $i$ is 
$\sum x_{1,j_1} \ldots x_{k,j_k} u_i(j_1,\ldots,j_k)$ 
where the sum is over all pure strategy profiles $(j_1,\ldots,j_k) \in S$.
 
A {\em Nash equilibrium} (NE) is a (mixed) strategy profile $x^*$ such that
all $i =1, \ldots,n$ and every mixed strategy 
$y_i$ for player $i$,
$U_i(x^*) \geq U_i(y_i;x^*_{-i})$. 
It is sufficient to check
switches to pure strategies only, i.e.,
$x^*$ is a NE iff $U_i(x^*) \geq U_i(\pi_{i,j};x^*_{-i})$ 
for every pure strategy $j \in S_i$, for
each player $i=1,\ldots,n$.
Every finite game has at least one NE \cite{Nash:1951:NCG}.

A mixed profile $x$ is called a {\em $\epsilon$-perfect equilibrium}
($\epsilon$-PE) if it is 
(a)  fully mixed, i.e., $x_{i,j} > 0$ for all $i$,
and (b),  for every player $i$ and pure strategy $j$, if $x_{i,j} > \epsilon$,
then the pure strategy $\pi_{i,j}$ is a best response for player $i$ 
to $x_{-i}$.
We call a mixed profile $x^*$,  a  {\em trembling hand perfect equilibrium} (PE) of $\Gamma$
if it is a limit point of $\epsilon$-PEs of the game $\Gamma$.  In other words,
we call $x$ a PE if there exists a sequence $\epsilon_k > 0$, such that 
$\lim_{k \rightarrow \infty} \epsilon_k = 0$, and such that  
for all $k$ there is a corresponding $\epsilon_k$-PE, $x^{\epsilon_k}$
of $\Gamma$, such that $\lim_{k \rightarrow \infty} x^{\epsilon_k} = x^*$.
Every finite game has at least one PE, and all PEs are NEs \cite{Selten:1975:PCEE}.

\subsection{Complexity theoretic notions}

A {\em $\{+,-,*,\max,\min\}$-circuit} is a circuit with inputs
$x_1, x_2, \ldots, x_n$, as well as rational constants,
and a finite number of (binary) computation gates taken from $\{+,-,*,\min,\max\}$, 
with a subset of the computation gates labeled $\{o_1, o_2, \ldots, o_m\}$ and called output 
gates.\footnote{ Note that the gates $\{+,-,*,\min,\max\}$ are of course redundant:
gates $\{+,*,\max\}$ with rational constants are equally expressive.}

All circuits of this paper are $\{+,-,*,\min,\max\}$-circuits, so we shall often just write ``circuit" for ``$\{+,-,*,\min,\max\}$-circuit". A circuit computes a continuous function from $\real^n \rightarrow \real^m$ (and $\rat^n \rightarrow \rat^m$) in the natural way. Abusing notation slightly, we shall often identify the circuit with the function it computes.

By a (total) {\em multi-valued function}, $f$, with domain $A$ and co-domain $B$,
we mean a function that maps each $a \in A$ to a non-empty subset 
$f(a) \subseteq B$.  
We use $f: A \twoheadrightarrow B$ to denote such a function. Intuitively, when considering a multi-valued function as a computational problem, we are interested in producing just one of the elements of $f(a)$ on input $a$, so we refer to $f(a)$ as the set of {\em allowed outputs}.
A multi-valued function  $f: \{0,1\}^* \twoheadrightarrow \real^*$ is said to be in $\FIXP$ if there is a polynomial time computable map, $r$, that maps each instance $I \in \{0,1\}^*$ of $f$ to $r(I)= \langle1^{k^I}, 1^{d^I}, P^I, C^I, a^I, b^I \rangle$, where
\begin{itemize}
\item{}$k^I, d^I$ are positive integers and $a^I, b^I \in \rat^{d^I}$.
\item{}$P^I$ is a convex polytope in $\real^{k^I}$, given as a set of 
linear inequalities with rational coefficients.
\item{}$C^I$ is a circuit which maps $P^I$ to itself.
\item{}$\phi^I: \{1,\ldots,d^I\} \rightarrow \{1, \ldots, k^I\} $ is a finite function given by its table.
\item{}$f(I) = \{(a^I_i y_{\phi^I(i)} + b^I_i)_{i=1}^{d^I} \mid y \in P^I \: \wedge \: C^I(y) = y\}$. 
Note that $f(I) \not = \emptyset$, by Brouwer's fixed point theorem.
\end{itemize}
The above is in fact one of many equivalent characterizations of $\FIXP$ \cite{EY07}.
Informally, $\FIXP$ are those real vector multi-valued functions, with discrete inputs, 
that can be cast as Brouwer fixed point computations.
A multi-valued function $f:  \{0,1\}^* \twoheadrightarrow \real^*$ is said to be $\FIXP${\em -complete} if:
\begin{enumerate}
\item{}$f \in \FIXP$, and
\item{}for all $g \in \FIXP$, there is a polynomial time computable map, mapping instances $I$ of $g$ to $\langle y^I, 1^{k^I}, 1^{d^I}, \phi^I, a^I, b^I \rangle$, where $y^I$ is an instance of $f$, $k^I$ and $d^I$ are positive integers, $\phi^I$ maps $\{1,\ldots, d^I\}$ to $\{1,\ldots,k^I\}$,  $a^I$ and $b^I$ are 
$d^I$-tuples with rational entries, so that 
$g(I) \supseteq  \{ (a^I_i z_{\phi^I(i)} + b^I_i)^{d^I}_{i=1} \mid 
z \in f(y^I) \}$.  In other words, for any allowed output $z$ of $f$ on input $y^I$, 
the vector $(a^I_i z_{\phi^I(i)} + b^I_i)^{d^I}_{i=1}$ is an allowed output of $g$ on input $I$. 
\end{enumerate}
Etessami and Yannakakis \cite{EY07} showed that the multi-valued function
which maps games in strategic form to their Nash equilibria is $\FIXP$-complete.\footnote{To view the Nash equilibrium problem as a total multi-valued function, $f_{\mbox{\rm \tiny Nash}}: \{0,1\}^* \twoheadrightarrow \real^*$, we can view all strings in $\{0,1\}^*$ as encoding some game, by viewing ``ill-formed" input strings as encoding a fixed trivial game.} 

Since the output of a $\FIXP$ function consists of real-valued vectors, and as there are 
circuits whose fixed points are all irrational, a $\FIXP$ function is not directly computable by a Turing machine, and the class is therefore not directly comparable with standard complexity classes of total search problems 
(such as \PPAD, \PLS, or \TFNP). 
This motivates the following definition of the discrete class $\FIXPA$, also from \cite{EY07}.
A multi-valued function $f: \{0,1\}^* \twoheadrightarrow \{0,1\}^*$ (a.k.a. a totally defined discrete search problem) is said to be in $\FIXPA$ if there is a function $f' \in \FIXP$,
and polynomial time computable maps $\delta: \{0,1\}^* \rightarrow \rat_+$ and $g: \{0,1\}^* \rightarrow \{0,1\}^*$, such that for all instances $I$, 
\[ f(I) \supseteq \{ \: g(\langle I, y \rangle) \mid  y \in \rat^* \: \wedge \:  
\exists y' \in f'(I): \:  \|y-y'\|_ \infty \leq \delta(I)  \: \}. \] 
Informally, $\FIXPA$ are those totally defined discrete search problems that reduce to approximating exact Brouwer fixed points.
A multi-valued function $f: \{0,1\}^* \twoheadrightarrow \{0,1\}^*$ is said to be $\FIXPA$-{\em complete} if:
\begin{enumerate}
\item{}$f \in \FIXPA$, and
\item{}For all $g \in \FIXPA$, there are polynomial time computable maps 
$r_1, r_2:\{0,1\}^* \rightarrow \{0,1\}^*$,
such that 
$g(I) \supseteq \{ \: r_2(\langle I,z \rangle)  \mid   z \in f(r_1(I)) \: \}$.
\end{enumerate}
Etessami and Yannakakis showed that the multi-valued function 
that maps pairs $\langle \Gamma, \delta \rangle$,
where $\Gamma$ is a strategic form game and $\delta > 0$, to the set of rational $\delta$-approximations   
(in $\ell_\infty$-distance) of Nash equilibria of $\Gamma$, is $\FIXPA$-complete.

\section{Computing $\epsilon$-PEs in \FIXP}
\label{Kousha}
Given a game $\Gamma$,
let $m = \sum_{i \in N} m_i$ denote the total number
of pure strategies of all players in $\Gamma$. For $\epsilon > 0$, 
let $D^{\epsilon}_{\Gamma} \subseteq D_{\Gamma}$ denote the polytope of fully mixed profiles
of $\Gamma$ such that furthermore every pure strategy is played with
probability at least $\epsilon > 0$ (recall that $D_\Gamma$ is the polytope of all strategy profiles). In this section, we show the following theorem.

\begin{theorem}
\label{fixp-no-division}
There is a function, $F^{\epsilon}_{\Gamma}(x): D_{\Gamma} \rightarrow D^{\epsilon}_{\Gamma}$,  given by a circuit
computable in polynomial time from $\Gamma$, with the circuit having both $x$ 
and $\epsilon > 0$ as its inputs, such that for all fixed $0 < \epsilon < 1/m$,
every Brouwer fixed point of the function $F^{\epsilon}_{\Gamma}(x)$ is an $\epsilon$-PE
of $\Gamma$. In particular, the problem of computing an $\epsilon$-perfect equilibrium 
for a finite $n$-player normal form game is in \FIXP.
\end{theorem}

The rest of the section is devoted to the proof of Theorem \ref{fixp-no-division}.
We will directly use, and somewhat modify, a construction
developed and used in \cite{EY07}  (Lemmas 4.6 and 4.7, and definitions
before them) which characterize
the Nash Equilibria of a game as fixed points of a $\{+,-,*,\max,\min\}$-circuit. In particular, compared to Nash's original functions \cite{Nash:1951:NCG}, the use of division is avoided.  The construction defined in \cite{EY07} that we modify 
amounts to a concrete algebraic realization of certain geometric 
characterizations of Nash Equilibria that were described by 
Gul, Pierce, and Stachetti in \cite{GPS}.

Concretely, suppose we are given $0 < \epsilon < 1/m$.
For each mixed strategy profile $x$, let $v(x)$
be a vector which gives the expected payoff of each pure strategy
of each player with respect to the profile $x$ for the other players.
That is, vector $x$ is a vector of dimension $m$,
whose entries are
indexed by pairs $(i,j), i=1,\ldots,n; j=1,\ldots,m_i$, and
$v(x)$ is also a vector of dimension $m$ whose $(i,j)$-entry is
$U_i(\pi_{i,j};x_{-i})$.  Let $h(x) = x+v(x)$.
We can write $h(x)$ as $(h_1(x), \ldots,h_n(x))$ where $h_i(x)$
is the subvector corresponding to the strategies of player $i$.
For each player $i$, consider the function 
$f_{i,x}(t) = \sum_{j \in S_i} \max (h_{ij}(x) -t,\epsilon)$.
Clearly, this is a continuous, piecewise linear function of $t$.
The function is strictly decreasing as $t$ ranges from $-\infty$
(where $f_{i,x}(t) = +\infty$) up to $\max_j h_{ij}(x) - \epsilon$ 
(where $f_{i,x}(t) = m_i \cdot \epsilon$).
Since we have $m_i \cdot \epsilon < 1$,
there is a unique value of $t$, call it $t_i$, where $f_{i,x}(t_i) =1$.
The function $F^{\epsilon}_{\Gamma}$ is defined as 
follows: 
\[ F^{\epsilon}_\Gamma(x)_{ij} = \max (h_{ij}(x) -t_i,\epsilon) \]
for every $i=1,\ldots,n$, and $j \in S_i$.
From our choice of  $t_i$, we have 
$\sum_{j \in S_i} F^{\epsilon}_{\Gamma}(x)_{ij} =1$ for all $i =1,\ldots,n$,
thus for any  mixed profile, $x$, we have 
$F^{\epsilon}_{\Gamma}(x) \in  D^{\epsilon}_{\Gamma}$.
So $F^{\epsilon}_{\Gamma}$ maps $D_\Gamma$ to $D^{\epsilon}_{\Gamma}$, and since it is clearly also continuous, it has fixed points, by Brouwer's theorem.

\begin{lemma} 
\label{new-nash}
For $0 < \epsilon < 1/m$, every fixed point of the function 
$F^{\epsilon}_{\Gamma}$
is an $\epsilon$-PE of $\Gamma$.
\end{lemma}
\begin{proof}
If $x$ is a fixed point of $F^{\epsilon}_\Gamma$, 
then $x_{ij}= \max (x_{ij}+ v(x)_{ij} -t_i,\epsilon)$
for all $i,j$. 
Recall that $v(x)_{ij}= U_i(\pi_{i,j};x_{-i})$ is the expected payoff for 
player $i$
of his $j$'th pure strategy $\pi_{i,j}$, with respect to strategies $x_{-i}$ of 
the other players.

Note that
the equation $x_{ij}= \max (x_{ij}+ U_i(\pi_{i,j};x_{-i})-t_i,\epsilon)$ 
implies that
$U_i(\pi_{i,j};x_{-i})= t_i$ for all $i,j$ such that $x_{ij}> \epsilon$,
and  that $U_i(\pi_{i,j};x_{-i}) \leq t_i$ for all $i,j$ such that 
$x_{ij}=\epsilon$.
Consequently, by definition, $x$ constitutes an $\epsilon$-PE.
\qed
\end{proof}

The following Lemma shows that
we can implement the function $F^{\epsilon}_\Gamma(x)$ by a
circuit
which has $x$ and $\epsilon$ as inputs.
The proof exploits sorting networks.

\begin{lemma}
\label{new-nash-circuit}
Given $\Gamma$,
we can construct in polynomial time a  $\{+,-,*,\max,\min \}$-circuit that computes the function $F^{\epsilon}_\Gamma(x)$,
where $x$ and $\epsilon$ are inputs to the circuit.
\end{lemma}
\begin{proof}
The circuit does the following.

Given a vector $x \in D_{\Gamma}$, 
first compute $y=h(x) = x+v(x)$.
It is clear from the definition of $v(x)$ that 
$y$ can be computed using $+,*$ gates.
For each player $i$, let $y_i$ be the corresponding subvector of $y$
induced by the strategies of player $i$. 
Sort $y_i$ in decreasing order, and let $z_i$
be the  resulting sorted vector, i.e. the components of 
$z_i=(z_{i1},\ldots,z_{im_i})$ are the same
as the components of $y_i$, but they are sorted: 
$z_{i1}\geq z_{i2} \geq \ldots \geq z_{im_i}$.
To obtain $z_i$, the circuit uses a polynomial sized sorting network, $W_i$, for each $i$ 
(see e.g. Knuth \cite{Knuth} for
background on sorting networks). For each comparator gate
of the sorting network we use a $\max$ and a $\min$ gate.

Using this, for each player $i$,
we compute $t_i$ as the following expression: 

\[\max \{ (1/l) \cdot  ( (\sum_{j=1}^l z_{ij}) +  (m_i - l) \cdot \epsilon - 1 ) | l=1,\cdots,m_i\}\]

We will show below that this expression does indeed give the correct value of 
$t_i$.
Finally, we output $x'_{ij} = \max(y_{ij}-t_i,\epsilon)$ for each $i=1,\ldots,d; j\in S_i$.

We now have to establish 
that $t_i= \max \{ (1/l) \cdot ( (\sum_{j=1}^l z_{ij}) + (m_i - l) \cdot \epsilon - 1 ) | l=1,\cdots,m_i\}$. 
Consider the function $f_{i,x}(t) = \sum_{j \in S_i} \max (z_{ij} -t,\epsilon)$ as $t$ decreases
from $z_{i1}-\epsilon$ where the function 
value is at its minimum of $m_i \epsilon$, down until the function reaches the value 1.
In the first interval from $z_{i1}-\epsilon$ to $z_{i2}-\epsilon$ the function is
$f_{i,x}(t) = z_{i1}-t + (m_i-1)\cdot\epsilon$; in the second interval from  $z_{i2}-\epsilon$ to 
$ z_{i3}-\epsilon$ it is
$f_{i,x}(t) = z_{i1}+z_{i2}-2t + (m_i-2)\cdot\epsilon$, and so forth.
In general, in the $l$-th interval, 
$f_{i,x}(t) = \sum_{j=1}^l (z_{ij}- t) + (m_i - l) \cdot \epsilon = \sum_{j=1}^l z_{ij}- lt + 
(m_i - l) \cdot \epsilon$.
If the function reaches the value 1 in the $l$'th interval, 
then clearly $t_i =  ((\sum_{j=1}^l z_{ij})  + (m_i - l) \cdot \epsilon -1)/l$.
In that case, furthermore
for $k<l$, we have $\sum_{j=1}^k (z_{ij} -t_i) + (m_i - k)\cdot \epsilon 
\leq \sum_{j=1}^l (z_{ij} -t_i) + (m_i - l)\cdot\epsilon = 1$, 
because in that case we know $(z_{ij}-t_i) \geq \epsilon$ for every $j \in \{1,\ldots,l\}$.
Therefore, in this case
$((\sum_{j=1}^k z_{ij}) + (m_i - k)\cdot \epsilon -1)/k \leq t_i$.
On the other hand, if $l < m_i$, then for $k > l$ we have $t_i \geq z_{ik}-\epsilon$, 
i.e., $z_{ik}- t_i \leq \epsilon$, and thus for all $k >l$, $k \leq m_i$,
we have
$\sum_{j=1}^k (z_{ij} -t_i) + (m_i - k)\cdot\epsilon 
\leq \sum_{j=1}^l (z_{ij} -t_i) + (m_i - l)\cdot\epsilon = 1$.
Thus again $((\sum_{j=1}^k z_{ij})+ (m_i -k)\cdot \epsilon -1)/k \leq t_i$.
Therefore, $t_i= \max \{ (1/l) \cdot  ( (\sum_{j=1}^l z_{ij}) + (m_i-l)\cdot \epsilon -1) | l=1,\cdots,m_i\}$. 
\qed
\end{proof}

Lemma \ref{new-nash} and Lemma \ref{new-nash-circuit} together immediately imply Theorem \ref{fixp-no-division}.

\section{Almost implies near}
\label{Kristoffer} 

As outlined in the introduction, 
in this section, we want to exploit the ``uniform'' function $F^{\epsilon}_{\Gamma}(x)$ 
devised in the previous section for $\epsilon$-PEs, and construct a 
``small enough''  $\epsilon^* > 0$ such that any fixed point of $F^{\epsilon^*}_{\Gamma}(x)$ is
$\delta$-close, for a given $\delta >0$, to an actual PE.

\newcommand{\transpose}{\ensuremath{\mathsf{T}}}
\newcommand{\abs}[1]{\ensuremath{\mathopen\lvert #1 \mathclose\rvert}}
\newcommand{\norm}[1]{\ensuremath{\mathopen\lVert #1 \mathclose\rVert}}
\newcommand{\Abs}[1]{\ensuremath{\left| #1 \right|}}
\newcommand{\Norm}[1]{\ensuremath{\left\| #1 \right\|}}
\newcommand{\RR}{\mathbb{R}}

\newcommand{\epsPE}{\ensuremath{\operatorname{EPS-PE}}}
\newcommand{\PE}{\ensuremath{\operatorname{PE}}}
\newcommand{\PEbound}{\ensuremath{\operatorname{PE-bound}_\delta}}

The following is a special case of the simple but powerful ``almost implies near'' paradigm of
Anderson~\cite{TAMS:Anderson86}.
\begin{lemma}[Almost implies near]
\label{LEM:AlmostNear}
  For any fixed strategic form game $\Gamma$, and any $\delta > 0$, there is an
  $\epsilon> 0$, so that any $\epsilon$-PE of $\Gamma$
  has $\ell_\infty$-distance at most $\delta$ to some PE of $\Gamma$.
\end{lemma}
\begin{proof}
  Assume to the contrary that there is a game $\Gamma$ and a $\delta > 0$
  so that for all $\epsilon > 0$, there is an $\epsilon$-PE $x_\epsilon$ of $\Gamma$ so that there is no PE in a $\delta$-neighborhood (with respect to $l_\infty$
  norm) of $x_\epsilon$. Consider the sequence $(x_{1/n})_{n \in {\bf
      N}}$. Since this is a sequence in a compact space (namely, the
  space of mixed strategy profiles of $\Gamma$), it has a limit point,
  $x^*$, which is a PE of $\Gamma$ (since $x_{\epsilon}$
  is a $\epsilon$-PE). But this contradicts the statement that
  there is no PE in a $\delta$-neighborhood of any of
  the profiles $x_{1/n}$. \qed
\end{proof}  

A priori, we have no bound on $\epsilon$, but we next use the machinery of real algebraic geometry \cite{BasuPollackRoy2006,BasuPollackRoy2011} to obtain a specific bound as a ``free lunch", just from the fact that Lemma \ref{LEM:AlmostNear} is true.
\begin{lemma}
\label{LEM:AlmostVeryNear}
There is a constant $c$, so that for all integers $n,m,k,B \in \nat$ and $\delta \in \rat_+$, 
the following holds. Let $\epsilon \leq \min(\delta,1/B)^{n^{c m^3}}$. For any $n$-player 
game $\Gamma$, with a combined total of $m$ pure strategies for all players, 
and with integer payoffs of absolute value at most $B$, any $\epsilon$-PE 
of $\Gamma$ has $\ell_\infty$-distance at most $\delta$ to some PE of $\Gamma$.
\end{lemma}
\begin{proof}  
The proof involves constructing formulas in the first order theory of real numbers, which 
formalize the ``almost implies near" statement of Lemma \ref{LEM:AlmostNear}, 
with $\delta$ being ``hardwired" as a constant and $\epsilon$ being the only free variable. Then, we apply {\em quantifier elimination} to these formulas. This leads to a quantifier free statement to which we can apply standard theorems bounding the size of an instantiation of the free variable $\epsilon$ making the formula true. We shall apply and refer to theorems in the monograph of Basu, Pollack and Roy \cite{BasuPollackRoy2006,BasuPollackRoy2011}. Note that we specifically refer to theorems and page numbers of the online edition \cite{BasuPollackRoy2011}; these are in general different from the printed edition \cite{BasuPollackRoy2006}.

\paragraph{First-order formula for an $\epsilon$-perfect equilibrium:} Define
$R_i(x \setminus k)$ as the polynomial expressing $U_i(\pi_{i,k};x_{-i})$, that is, the expected payoff to player
$i$ when it uses pure strategy $k$, and the other players play according
to their mixed strategy in the profile $x$. Thus,
\[
R_i(x \setminus k) := \sum_{a_{-i}} u_i(k;a_{-i}) \prod_{j\neq i} x_{j,a_j}.
\]

Let $\epsPE(x,\epsilon)$ be the quantifier-free first-order formula, with free
variables $x\in\RR^m$ and $\epsilon\in\RR$, defined by the conjunction
of the following formulas that together express that $x$ is
an $\epsilon$-perfect equilibrium:
\begin{gather*}
x_{i,j} > 0 \quad \text{for } i=1\dots,n \text{, and } j=1,\dots,m_i \enspace ,\\
x_{i,1}+ \dots + x_{i,m_i} = 1 \quad \text{for } i=1\dots,n \enspace ,\\
\left(R_i(x \setminus k) \geq R_i(x \setminus l)\right) \vee \left(x_{i,k} \leq \epsilon \right) \quad \text{for } i=1\dots,n \text{, and } k,l=1,\dots,m_i \enspace .
\end{gather*}

\paragraph{First-order formula for  perfect equilibrium:}
Let $\PE(x)$ denote the following first-order formula with free
variables $x\in \RR^m$, expressing that $x$ is a perfect equilibrium:
\begin{gather*}
\forall \epsilon >0 \: \exists y \in \RR^m : \epsPE(y,\epsilon) \wedge \norm{x-y}^2 < \epsilon \enspace .
\end{gather*}

\paragraph{First-order formula for ``almost implies near'' statement:}

Given a {\em fixed} $\delta>0$ let $\PEbound(\epsilon)$ denote the following
first-order formula with free variable $\epsilon\in\RR$, denoting that
any $\epsilon$-perfect equilibrium of $G$ is $\delta$-close to a
perfect equilibrium (in $\ell_2$-distance, and therefore also in $\ell_\infty$-distance):

\begin{gather*}
\forall x \in \RR^m \: \exists y \in \RR^m : (\epsilon>0) \wedge \left(\neg \epsPE(x,\epsilon) \vee \left(\PE(y) \wedge \norm{x-y}^2 < \delta^2 \right)\right) \enspace .
\end{gather*}

Suppose $\delta^2 = 2^{-k}$ and the payoffs have absolute value at most $B=2^\tau$. Then for this formula we have
\begin{itemize}
\item The total degree of all involved polynomials is at most $\max(2,n-1)$.
\item The bitsize of coefficients is at most $\max(k,\tau)$.
\item The number of free variables is $1$.
\item Converting to prenex normal form, the formula has 4 blocks of
  quantifiers, of sizes $m$, $m$, $1$, $m$, respectively.
\end{itemize}
 
We now apply quantifier elimination
\cite[Algorithm 14.6, page 555]{BasuPollackRoy2011} to the formula $\PEbound(\epsilon)$, converting it into an
equivalent quantifier free formula $\PEbound'(\epsilon)$ with a single free variable $\epsilon$. This is simply a Boolean formula whose atoms are sign conditions on various polynomials in $\epsilon$. The bounds given by Basu, Pollack and Roy in association to Algorithm 14.6 imply that for this formula:
\begin{itemize}
\item The degree of all involved polynomials (which are univariate polynomials in $\epsilon$) is $\max(2,n-1)^{O(m^3)} = n^{O(m^3)}$.
\item The bitsize of all coefficients is at most $\max(k,\tau)\max(2,n-1)^{O(m^3)} =
  \max(k,\tau)n^{O(m^3)}$.
\end{itemize}

By Lemma~\ref{LEM:AlmostNear}, we know that there exists an $\epsilon>0$
so that the formula $\PEbound'(\epsilon)$ is true. We now apply (as the involved polynomials are univariate, simpler tools would also suffice) Theorem 13.14 of Basu, Pollack and Roy \cite[Page 521]{BasuPollackRoy2011} 
to the set of polynomials that are atoms of $\PEbound'(\epsilon)$  and conclude that $\PEbound'(\epsilon^*)$ is true for some
$\epsilon^* \geq 2^{-\max(k,\tau)n^{O(m^3)}}$. By the semantics of the formula $\PEbound(\epsilon)$, we also have that $\PEbound(\epsilon')$ is true for all $\epsilon' \leq \epsilon^*$, and the statement of the lemma follows. \qed 
\end{proof}

\section{Proof of the main theorem}
\label{wrapup}
We now prove Theorem \ref{thm-main}. Let $\Gamma$ be the $n$-player game given as input. Let $m$ be the combined total number of pure strategies for all players. 
Let $B \in \nat$ be the largest absolute value of any payoff of $\Gamma$. By the definition of $\FIXPA$, our task is the following. Given a parameter $\delta > 0$, we must construct a polytope $P$, a circuit $C: P \rightarrow P$, and a number $\delta'$,  so that $\delta'$-approximations to fixed points of $C$ can be efficiently transformed into $\delta$-approximations of PEs of $\Gamma$. In fact, we shall let $\delta' = \delta/2$ and ensure that $\delta'$-approximations to fixed points of $C$ {\em are} $\delta$-approximations of PEs of $\Gamma$. The polytope $P$ is simply the polytope $D_\Gamma$ of all strategy profiles of $\Gamma$; clearly we can output the inequalities defining this polytope in polynomial time. The circuit $C$ is the following: We construct the circuit for the function $F^\epsilon_\Gamma$ of Section \ref{Kousha}. Then, we construct a circuit for the number $\epsilon^* = \min(\delta/2, B^{-1})^{{2^{\lceil c m^3 \lg n \rceil }}} \leq \min(\delta/2,B^{-1})^{n^{c m^3}}$, where $c$ is the constant of Lemma \ref{LEM:AlmostVeryNear}: The circuit simply repeatedly squares the number $\min(\delta/2, B^{-1})$ (which is a rational constant) and thereby consists of exactly $\lceil c m^3 \lg n \rceil$ multiplication gates, i.e., a polynomially bounded number. We then plug in the circuit for $\epsilon^*$ for the parameter $\epsilon$ in the circuit for $F^\epsilon_\Gamma$, obtaining the circuit $C$, which is obviously a circuit for $F^{\epsilon^*}_\Gamma$. Now, by Theorem \ref{fixp-no-division}, any fixed point of $C$ on $P$ is an $\epsilon^*$-PE of $\Gamma$. Therefore, by Lemma \ref{LEM:AlmostVeryNear}, any fixed point of $C$ is a $\delta/2$-approximation in $\ell_\infty$-distance to a PE of $\Gamma$. Finally, by the triangle inequality, any $\delta'=\delta/2$-approximation to a fixed point of $C$ on $P$ is a $\delta/2 + \delta/2 = \delta$ approximation of a PE of $\Gamma$. This completes the proof.
\section{Conclusion}

We have showed that the problem
of {\em approximating} a  trembling hand perfect equilibrium 
for a finite strategic form game is in \FIXPA. We do not know if {\em exactly} computing a trembling hand perfect equilibrium is in \FIXP, and we consider this an interesting open problem, although it should be noted that if one is interested exclusively in the Turing Machine complexity of the problem, $\FIXPA$ membership of the approximation version is arguably ``the real thing". We also note that this makes our proof interesting as a case where membership in \FIXPA\ is {\em not} established as a simple corollary of the exact problem being in the ``abstract class" \FIXP, as was the case for all examples in the original paper of  Etessami and Yannakakis.
 
 \subsection*{Acknowledgements}
Hansen and Miltersen acknowledge support from the Danish National Research Foundation and The National Science Foundation of China (under the grant 61061130540) for the Sino-Danish Center for the Theory of Interactive Computation,
within which their work was performed. They also acknowledge support from the Center for Research
in Foundations of Electronic Markets (CFEM), supported by the Danish Strategic Research Council.

\bibliographystyle{plain}

\end{document}